%% file: main.tex
\documentclass[conference,letterpaper]{IEEEtran}
\input{PackageHeader2022.tex}

\input{Notations22}
\input{acronyms}

\def\BibTeX{{\rm B\kern-.05em{\sc i\kern-.025em b}\kern-.08em T\kern-.1667em\lower.7ex\hbox{E}\kern-.125emX}}
\settoggle{Track}{true}

\newcommand{\addvR}[1]{\addv{a}{off}{#1}}
\def\citep{\cite}
\def\citet{\cite}

\begin{document}
\title{New Bounds on Quantum Sample Complexity of Measurement Classes}

\author{
\IEEEauthorblockN{Mohsen Heidari}
\IEEEauthorblockA{Department of Computer Science\\Indiana University\\
Email: mheidar@iu.edu}
\and
\IEEEauthorblockN{Wojciech Szpankowski}
\IEEEauthorblockA{Department Computer Science\\
Purdue University\\
Email: szpan@purdue.edu}
    }
\maketitle
\begin{abstract}
This paper studies quantum supervised learning for classical inference from quantum states. In this model, a learner has access to a set of labeled quantum samples as the training set. The objective is to find a quantum measurement that predicts the label of the unseen samples. The hardness of learning is measured via sample complexity under a quantum counterpart of the well-known \ac{PAC}. Quantum sample complexity is expected to be higher than classical one, because of the measurement incompatibility and state collapse. Recent efforts showed that the sample complexity of learning a finite quantum concept class $\mathcal{C}$ scales as $O(\abs{\mathcal{C}})$.  This is significantly higher than the classical sample complexity that grows logarithmically with the class size.      This work improves the sample complexity bound to   $O(V_{\mathcal{C}^*}\log |\mathcal{C}^*|)$, where $\mathcal{C}^*$ is the set of extreme points of the convex closure of $\mathcal{C}$ and $V_{\mathcal{C}^*}$ is the \textit{shadow-norm} of this set.  We show the tightness of our bound for the class of bounded Hilbert–Schmidt norm, scaling as $O(\log |\mathcal{C}^*|)$.  Our approach is based on a new quantum \ac{ERM} algorithm equipped with a shadow tomography method. 
\end{abstract}

\section{Introduction}
Quantum learning is one of the leading applications of quantum computing both for classical and quantum problems. While some models suggest quantum-enhancements of classical learning by mapping data into input quantum states \cite{Giovannetti2008,Park2019,Lloyd2013,Lloyd2014}, \acp{QC} have a far greater capability to learn patterns from inherently quantum data. This is possible by directly operating on quantum states of physical systems (e.g., photons or states of matter) or their qubit representations \cite{Carrasquilla2017,Carleo2017,Broughton2020,massoli2021leap,Lu2018,Kassal2011,McArdle2020,Cao2019}. Learning from quantum data has been studied extensively in recent literature in the context of diverse applications, including quantum simulation \cite{Kassal2011,McArdle2020,Hempel2018,Cao2019,Bauer2020}, phase-of-matter detection \cite{Carrasquilla2017,Broecker2017}, ground-state search \cite{Carleo2017,Broughton2020,Biamonte2017}, and entanglement detection \cite{Ma2018,massoli2021leap,Lu2018,Hiesmayr2021,Chen2021a,Deng2017}. In this context, \acp{VQA} and \acp{QNN} are among the most conventional models of a quantum learner that consist of a parameterized circuit trained in a hybrid quantum-classical loop. 

One of the objectives of quantum learning is to train a model with a minimum number of iterations or samples. This is crucial in scenarios where quantum data is scarce or is expensive to produce. For instance, sometimes in the quantum simulation of ground state search, the quantum states are the output of a Hamiltonian with high gate complexity, making the state preparation expensive.  An explicit example is the Hamiltonian of the two-dimensional Fermi-Hubbard model on an $8 \times 8$ lattice that requires about $10^7$ Toffoli gates \cite{Kivlichan2020}. The challenge is exacerbated as the training models are data-intensive due to the \textit{no-cloning} and \textit{state collapse} that prohibit sample reuse --- a phenomenon unique to quantum. Therefore, studying the quantum sample complexity, as a measure of the hardness of training a model, is foundational to our understanding of quantum learning. This is the focus of our study. 

In classical learning theory,  sample complexity has been studied extensively for decades under the well-known \ac{PAC} framework \citep{Kearns1994,Valiant1984}. The theory of quantum learning has been studied recently under various models \cite{Aaronson2007,Barnett2009,Bshouty1998,Badescu2019,Montanaro2016,HeidariQuantum2021,Heidari2023Quantum1}.
A relevant example is \textit{quantum state discrimination}, where   the objective is to distinguish an unknown quantum state $\rho$ from another (known or unknown) state using \textit{measurements} on multiple samples \cite{ODonnell2016,Haah2016,Badescu2019,Bubeck2020,Barnett2009,Gambs2008,Guta2010,Cheng2021}. In the Bshouty and Jackson \cite{Bshouty1998}, a classical PAC problem is studied using a quantum oracle that outputs identical copies of an associated superposition state \cite{Bshouty1998,Arunachalam2017,Kanade2018,Bernstein1997,Servedio2004}.

In this work, we consider \ac{QPAC} which is a natural generalization of PAC in quantum \cite{HeidariQuantum2021}. QPAC applies to conventional quantum supervised learning with VQAs, where one aims to accurately predict a
classical property (a.k.a., label) of an unknown quantum
state.
QPAC is a generic framework that subsumes several models such as state discrimination, quantum property testing,  quantum state classification, and classical PAC. A learner has access to $n$ physical quantum states with their classical label as the supervised training set. The learner's objective is to find a quantum measurement that minimizes the loss for predicting the label of unseen samples. 
Here, the learner does not know the sample distribution, the classical description of the states, and  the labeling law. It can only measure the $n$ physical samples. With the \textit{no-free lunch} theorem, the learner can only hope to  find a predictor with a loss sufficiently close to the best candidate in  a predefined set called the \textit{concept class}. \acp{QNN} or variational circuits are examples of concept classes. 
The quantum sample complexity of a given concept class is the minimum sample size $n$ such that a learner successfully finds a predictor performing close to the best candidate of the class. 


QPAC  is therefore a distribution-free and state-free requirement. It is stronger than PAC, as PAC is only distribution-free.  QPAC abides by quantum mechanical laws such as no-cloning, state collapse, and measurement incompatibility. Such properties prohibit sample reuse and, thus,  raise new challenges for learning in quantum settings. Moreover, quantum models (i.e., \acp{QNN}) are significantly  richer than classical models. Therefore, given the fragility of quantum samples, the strictness of QPAC, and the richness of quantum models, one expects  quantum sample complexity to be significantly, if not exponentially, greater than the classical one. 



Recent results show that any finite concept class $\mathcal{C}$ is learnable with $\order{|\mathcal{C}|}$ quantum samples  \cite{HeidariQuantum2021} leaving an exponential gap compared to classical sample complexity that scales with $\order{\log |\mathcal{C}|}$.  This gap is mainly due to the no-cloning and state collapse of quantum. However, a better bound is remained to be shown. Infinite classes have been studied recently in \cite{Padakandla2022,Heidari2023Quantum1}. In \cite{Padakandla2022} the results of  \cite{HeidariQuantum2019}  was extended to infinite concept classes through an $\epsilon$-netting argument. The authors in \cite{Heidari2023Quantum1} studied learning of $k$-qubit (or $k$-local) operators i.e., operations acting non-trivially on at most $k$ out of $d$ qubits.  They showed that the sample complexity scales as $\order{k~4^k \log d}$. When $k\ll d$, this bound is comparable to $\order{k\log d}$ which is the sample complexity of classical $k$-juntas.

\subsection{Main Contributions}
We study learning of generic finite or infinite quantum concept classes. Our main contributions are two-fold. 

\noindent\textbf{Bounds on quantum sample complexity.} This paper presents a new improved bound on the quantum sample complexity of generic concept classes (finite or infinite). Particularly, we prove that the quantum sample complexity of any measurement concept class $\mathcal{C}$ grows with $\order{V_\Cextreme\log |\mathcal{C}^*|}$, where $\mathcal{C}^*$ is the set of extreme points of the convex closure of the concept class and $V_\Cextreme$ is the \textit{shadow-norm} of this set. In general, $|\Cextreme|\leq |\mathcal{C}|$. In an extreme case $ |\mathcal{C}|$ could be infinite while $|\Cextreme|$ is finite. The bound is tight and simplifies to $\order{\log |\mathcal{C}|}$ for concept classed with bounded Hilbert–Schmidt norm.  
These results are surprising especially since sample duplication is prohibited and measurement incompatibility would lead to an exponentially larger sample complexity with standard methods. Such results are steps toward a quantum analogous to the fundamental theorem of statistical learning, which remains open.  

\noindent\textbf{QSRM algorithm.} We propose a new approach called \ac{QSRM} to measure the empirical risk of the predictors in the class  that substantially improves the quantum sample complexity of existing works.  Our approach relies on searching only the \textit{extreme} points of the concept class and making use of shadow tomography \cite{Huang2020}.  
In classical learning \ac{ERM} is a brute-force search to minimize the empirical loss and has a sample complexity of $\order{\log |\clC|}$. Extending this algorithm to quantum is not straightforward. One could naively propose the same technique to compute the empirical risk of each quantum predictor and choose the one with the minimum risk. However, the no-cloning and measurement incompatibility makes this approach prohibitive. Essentially, the training samples will be distorted each time we measure the empirical risk of a predictor. If we use that naively, then we might need fresh samples for each predictor, leading to a sample complicity of $\order{|\clC|}$. There have been multiple attempts \cite{HeidariQuantum2021,Padakandla2022} to introduce an efficient quantum ERM, obtaining a better bound that grows with $O(\log |\clC|)$ when measurements in $\clC$ are fully \textit{compatible} and  $O(|\clC|)$ in fully \textit{incompatible} class. Our algorithm achieves a sample complexity $O(V_\Cextreme\log |\mathcal{C}^*|)$ that improves upon these works.


\vspace{-0.15cm}
\section{Model Formulation}

\noindent\textbf{Notations:}  For shorthand, denote $[d]$ as $\set{1,2,...,d}$. 
As usual, a quantum measurement $\mathcal{M}$ is  a \ac{POVM} represented by a set of operators $\mathcal{M}:=\{M_v, v\in\mathcal{V}\}$, where $\mathcal{V}$ is the set of possible outcomes, $M_v\geq 0$ for any $v\in \mathcal{V}$, and $\sum_{v\in \mathcal{V}} M_v =I.$ Throughout of this paper, we assume a finite-dimensional Hilbert space $\mathcal{H}$ with dimensionality denoted by $\dimH$.
 %


\subsection{Quantum Learning Model}
Our model follows that of \citep{HeidariQuantum2021}. The objective is to  distinguish between multiple groups of unknown quantum states without prior knowledge about the states. Available is only a training set of quantum states with a classical label determining their group index. We seek a  procedure that  learns the labeling law. The model is defined more precisely as follows. 

Let $\mathcal{Y}$ denote the labeling set and $\mathcal{H}$ be the underlying finite-dimensional Hilbert space. There are $n$ sample $( \ket{\phi}_i, y_i), i\in [n]$ randomly generated according to an unknown but fixed probability distribution $D$\footnote{Alternatively, one can view the samples as copies of the mixed state $\rho_{XY}:=\EE_D[\ketbra{\phi}\tensor\ketbra{Y}]$.}.  A predictor is a quantum measurement $\mathcal{M}:=\set{M_{\hat{y}}:  \hat{y}\in \mathcal{Y}}$ that acts on the quantum states and outputs $\hat{y}\in \mathcal{Y}$ as the predicted label. 
From Born's rule, $\hat{y}_i$ is generated randomly with probability $\matrixelement{\phi_i}{M_{\hat{y}_i}}{\phi_i}.$   The prediction loss is  determined via a loss function $l: \mathcal{Y}\times \mathcal{Y}\rightarrow [0, \infty)$. 
   Therefore, the generalization (expected) loss is calculated as  $\Loss(\mathcal{M}) =\EE[l(Y,\hat{Y})],$ where the expectation is taken over all the randomness involved. 
   
   The learner's objective is to minimize $\Loss$ while having no prior knowledge about $D$ and the structure of the states. It  can only measure the $n$ physical samples without  knowing their classical descriptions. Due to the \textit{no-free lunch} theorem, the learner can only hope to find a $\mathcal{M}$ with $\Loss(\mathcal{M})$  sufficiently close to the best candidate within a fixed collection  $\mathcal{C}$ of predictors called the \textit{concept class}. 

\begin{definition}[QPAC] \label{def:quantum PAC}
A quantum learning algorithm \textit{agnostically} QPAC learns a concept class $\mathcal{C}$  if there exists a function $n_{\mathcal{C}}: (0,1)^2\mapsto \NN$ such that for every $\epsilon, \delta \in [0,1]$ and given $n>n_{\mathcal{C}}(\epsilon,\delta)$  samples drawn \ac{iid} according to any probability distributions $D$ and any unknown states $\ket{\phi}_i, i\in [n]$, the algorithm outputs, with probability of at least $(1-\delta)$, a measurement  whose loss is less than $\inf_{\mathcal{M}\in \mathcal{C}} \Loss(\mathcal{M})+\epsilon$.\footnote{Naturally, we are interested in efficient learning where   $n_{\mathcal{C}}$ is polynomial in $\epsilon, \delta$ and $\dimH$.} The quantum sample complexity of  $\mathcal{C}$ is the minimum of $n_{\mathcal{C}}$ for which there exists a QPAC learner.
\end{definition}

The state discrimination problem is a special case in which samples are identical and are either of two \textit{a priori} known states.  QPAC also subsumes classical PAC as classical samples and concept classes are mutually diagonalized in the canonical basis. Generally, QPAC is a stronger requirement than PAC and other methods, as it is a \textit{agnostic}, distribution-free and state-free condition. Hence, one expects that quantum sample complexity is significantly higher than classical one.   

\vspace{-0.1cm}
\subsection{Related Works on QERM}
It is known that classical \ac{ERM} PAC learns any (classical) finite concept class $\mathcal{C}$ with sample complexity that scales with $\order{\frac{1}{\epsilon^2} \log \frac{|\mathcal{C}|}{\delta}}$. In quantum settings, there have been various attempts in developing counterparts of \ac{ERM} algorithm. Due to the no-cloning theorem, the straightforward quantum extension of \ac{ERM} results in a sample complexity of $\order{\frac{|\mathcal{C}|}{\epsilon^2} \log \frac{1}{\delta}}$, see \cite{HeidariQuantum2021} for more details. This is problematic as the sample complexity grows linearly with the size of the concept class.  

In \cite{HeidariQuantum2021}, a new ERM-type algorithm is introduced to improve this bound.  The new bound depends on the measurement \textit{incompatibility} structure of the concept class. Incompatible measurements cannot be measured simultaneously (for more details see \cite{Holevo2012}). 
On one extreme, all the measurements in the concept class are mutually compatible; on another extreme, there is no pair of compatible measurements. Based on this, an improved bound on sample complexity is as follows.

\begin{fact}[\citep{HeidariQuantum2021}]\label{rem:QERM ISIT}
Quantum sample complexity of any finite concept class  $\mathcal{C}$ is upper bounded as
\begin{align*}
n_{\mathcal{C}}(\epsilon, \delta) \leq  \min_{\mathcal{C}_r \text{Comp. partition}} \sum_{r=1}^b \Big\lceil{\frac{8}{\epsilon^2}\log\frac{2b|\mathcal{C}_r|}{\delta}}\Big\rceil,
\end{align*}
where the minimization is taken over all compatibility partitioning of $\mathcal{C}$ \addvR{and $b$ is the number of partition bins}. This bound ranges from $\order{\frac{1}{\epsilon^2} \log \frac{|\mathcal{C}|}{\delta}}$, for fully compatible class, to  $\order{\frac{|\mathcal{C}|}{\epsilon^2} \log \frac{1}{\delta}}$ for fully incompatible class.
\end{fact}
In \cite{Padakandla2022}, this result was extended to infinite concept classes through an $\epsilon$-netting argument.  
In this paper, we propose a new quantum ERM that substantially improves the above bounds to one that, in the worst case, grows with $\order{\frac{1}{\epsilon^2} \log \frac{|\mathcal{C}|}{\delta}}$ even for fully incompatible concept classes. 

\section{Main Results}
In this section, we present the main results of the paper on the sample complexity. We prove the upper bound by proposing a quantum shadow risk minimization (QSRM). We first discuss methods to measure the empirical loss. 

\subsection{Measuring the Empirical Loss}
Without loss of generality, assume $l: \mathcal{Y}\times \mathcal{Y} \mapsto [0,1]$.
 Let $\mathcal{Z}$ be the image set of $l$, that is $\mathcal{Z}:= \set{l(y,\hat{y}): y, \hat{y}\in \mathcal{Y} }$. Since $\mathcal{Y}$ is a finite set, then so is $\mathcal{Z}$. We define a loss observable for any predictor $\mathcal{M}$ 
 as
$\MLoss:=\set{L^M_z: z\in\mathcal{Z}}$, where  
\begin{align}\label{eq:Loss operators}
 L^M_z = \sum_{\substack{y, \hat{y}\in \mathcal{Y}}} \11_{\set{l(y, \hat{y})=z}} M_{\hat{y}}\tensor \ketbra{y}, \qquad  \forall z\in\mathcal{Z}.
 \end{align} 
With that, the expectation value of $\MLoss$ is equal to $\Loss(\mathcal{M})$. Moreover, applying $\MLoss$ on any sample gives the incurred empirical loss of $\mathcal{M}$ on that sample. 

Now given a set of samples $\mathcal{S}_n = \set{ (\ket{\phi_i}, y_i): i\in [n]}$, we would like to estimate the empirical loss for all predictors in $\clC$. In the classical setting, this can be done easily by testing each predictor on the training samples, leading to a $\order{\log |\clC|}$ bound on $n$.  However, in quantum, the samples can be used once as they collapse after the measurements. 
Moreover,  the loss measurements $\MLoss$ might be \textit{incompatible} for different $\mathcal{M}\in \mathcal{C}$ --- rendering simultaneous measuring obsolete. As a result, a naive ERM in quantum would need fresh copies for each predictor, needing $\order{|\clC|}$ samples.  

In what follows, we introduce a new approach for measuring the empirical loss.
Our approach is inspired by shadow tomography \citep{Huang2020} that is applied to {\it identical copies} of quantum states.  Shadow tomography is a procedure to obtain useful classical information via random measuring of the input state. In what follows we explain this procedure.  
 
\addvR{Consider a generic state $\rho$ in a Hilbert space $\mathcal{H}$.  Generate a  unitary operator $U$ randomly from a class of choices $\mathcal{U}$ to be determined.  
Apply $U$ on the input state resulting in the state $U^\dagger \rho U$. Next, measure the rotated state in the canonical basis $\ket{j}, j \in [\dimH]$. From Born’s rule the probability of getting  the output $j$ is $p_{j}= \matrixelement{j}{U^{\dagger}\rho U}{j}$. 
Given an outcome $j$, prepare the state $\omega_j = U\ketbra{j}U^\dagger.$ 
The expectation   $\EE_{\sim (J, U)}[\omega_J]$  over the measurement randomness ($p_j$) and the choice of  $U$ equals to $\Gamma[\rho],$ where $\Gamma$ is a mapping defined as
\begin{align}\label{eq:Gamma}
\Gamma[O]:=\EE_{U}\Big[\sum_{j\in [\dimH]} \matrixelement{j}{U^{\dagger}O U}{j}~ U\ketbra{j}U^\dagger\Big],
\end{align}
for any operator $O$ on $\mathcal{H}$. 
Observe that $\Gamma$ is a linear mapping on $\mathcal{B}(\mathcal{H})$ \addvR{and hence has an inverse denoted by $\Gamma^{-1}$. We note that $\Gamma^{-1}$ is the shadow channel $\mathcal{M}^{-1}$ introduced in \cite{Huang2020}.} 
We apply $\Gamma^{-1}$ on $\omega_j$ resulting in the so called shadow
\begin{align}\label{eq:rho hat shadow}
\hat{\rho}:=\Gamma^{-1}\big[U\ketbra{j}U^\dagger\big].
\end{align}
Note that $\hat{\rho}$ is a classical matrix and hence can be copied several times. Moreover, $\hat{\rho}$ is not a valid density operator as it is not necessarily a positive semi-definite matrix. However, it is an unbiased estimate of the original state. 
\begin{fact}[\cite{Huang2020}]\label{fact:hat rho is unbiased}
When $\mathcal{U}$ is tomographically complete, the classical shadow $\hat{\rho}$ is unbiased, that is $\EE_{U, J}[\hat{\rho}]=\rho$. 
\end{fact}
For our problem, we apply the above process with a random unitary $U_i$ for each sample $\rho_i:=\ketbra{\phi_i}$ resulting in $n$ classical shadows $\hat{\rho}_i, i\in [n]$.} 
Next, we compute the expected loss per shadow as 
\begin{align*}
 \hat{L}_i(M)=  \sum_{\hat{y}} l(y, \hat{y}) \tr{M_{\hat{y}} \hat{\rho}_i}.
\end{align*}
This process is demonstrated in Figure \ref{fig:Estimation}.  Next, we average  it over all the shadows to find 
\begin{align}\label{eq:shadow loss}
L_{\hat{\mathcal{S}}_n}(\mathcal{M}) := \frac{1}{n}\sum_{i=1}^n \hat{L}_i(M), 
\end{align}
where  $\hat{\mathcal{S}}_n$ represents the set of all the shadows $\hat{\rho}_i$ with their labels $y_i$. Since the classical shadows can be copied,  this procedure estimates the empirical loss of all the predictors in the concept class without asking for fresh samples as in the naive approach. 
\begin{remark}
    The estimate $L_{\hat{\mathcal{S}}_n}$ is unbiased for any predictor $\mathcal{M}$, that is $\EE[L_{\hat{\mathcal{S}}_n}] = \Loss(\mathcal{M})$.
\end{remark}
\begin{remark}
    Huang \etal \cite{Huang2020} used the \textit{median of means} estimator for predicting the expectation of arbitrary observables. In our work, we slightly deviate and consider the empirical mean estimator to make sure our sample complexity bounds have expressions comparable with standard PAC bounds. 
\end{remark}
\begin{figure}[tbp]
\centering
\includegraphics[width=0.5\textwidth]{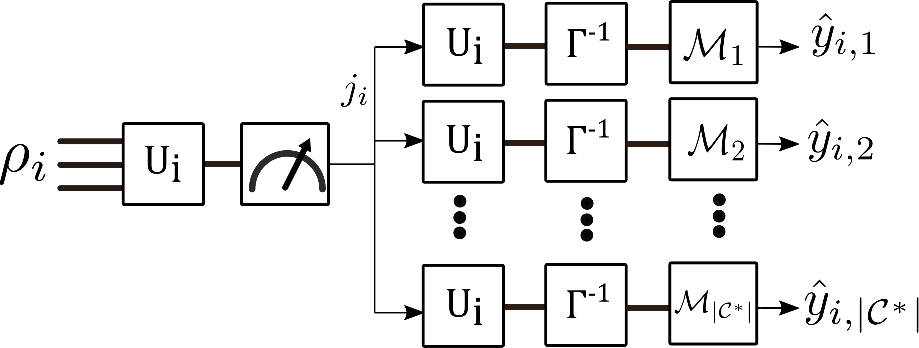}
\caption{The process for estimating the empirical loss of each measurement.    }
\label{fig:Estimation}
\end{figure}

\subsection{Concentration bounds}

We prove that a similar result as in \cite{Huang2020} also holds for the empirical average estimator, as opposed to the median of means. We proceed with presenting the following definition.
\begin{definition} \label{def:shadow norm}
The shadow norm of any operator $O$ on $\mathcal{H}$ is 
\small
\begin{align*}
\norm{O}_{\emph{shadow}}: = \max_{\sigma\in \mathcal{D}[H]}\Big(\sum_{j\in [\dimH]}\hspace{-10pt} \matrixelement{j}{U^{\dagger}\sigma U}{j}~ \expval{U\Gamma^{-1}[O]U^\dagger}{j}^2\Big)^{1/2}. 
\end{align*}
\normalsize
\end{definition}
\addvR{Depending on the choice of $\mathcal{U}$ different bounds are obtained for shadow norm. When $\mathcal{U}$ is the Clifford group, Shadow norm  is related to the Hilbert-Schmidt norm  }\cite{Huang2020}:  
\begin{equation}\label{eq:shadow norm bound}
    \norm{O}_{\emph{shadow}} \leq \sqrt{3\tr{O^2}}. 
\end{equation}
When we use random Pauli measurements, the shadow norm relates to the locality of an operator. More precisely, for an operator $O$ that acts non-trivially on only $k$ out of $d$ qubits, the shadow norm is bounded as $\norm{O}_{\emph{shadow}} \leq 2^k \norm{O}_\infty$.  Our first result extends Theorem 1 of \cite{Huang2020} to empirical mean estimators.

\begin{proposition}\label{prop:shadow empirical estimator}
Let $\hat{\rho}_i, i\in [n]$  be the classical shadows of $n$ copies of a mixed state $\rho$, as in \eqref{eq:rho hat shadow}. Then, given $\delta \in (0,1)$ and  $m$ arbitrary observables $O_1, \cdots, O_m$,  the empirical average  $\hat{o}_j := \frac{1}{n}\sum_{i}\tr{O_j\hat{\rho}_i}$ satisfies the additive error $\varepsilon$ given that 
\[n \geq \order{\frac{1}{\varepsilon^2}\log(\frac{m}{\delta}) \max_i \norm{O_i-\frac{\tr{O_i}}{\dimH} I}^2_{\emph{shadow}}}.\]
\end{proposition}
\addvR{The proof is given in Section \ref{subsec:proof prop}. The proposition implies that the median-of-means can be replaced with an empirical averaging without performance loss. Similar observations have been made regarding 3-design circuits and the entire unitary
group for classical shadow tomography \cite{Helsen2023,Zhao2021}. } 
Recall that the objective of \ac{QSRM} is to find $\min_{\mathcal{M}} L_{\hat{\mathcal{S}}_n}.$ However, the upper bound in the above theorem becomes loose when the concept class is large, for example when $|\mathcal{C}|$ is infinite.  Next, We show that only the extreme points of $\mathcal{C}$ are relevant. 
\subsection{Extreme points of a concept class}
Let $\Cconvex$ denote the convex closure (envelope) of $\mathcal{C}$. Note that $\Cconvex$ is the set of all POVMs that can be written as a convex combination of measurements in $\mathcal{C}$. More precisely, POVMs of the form $\bar{\mathcal{M}}=\set{\bar{M}_{\hat{y}}, \hat{y}\in \mathcal{Y}}$ such that 
\begin{equation*} 
\bar{M}_{\hat{y}} = \sum_{j=1}^k \alpha_j M^j_{\hat{y}},\quad \forall \hat{y}\in \mathcal{Y},
\end{equation*}
where each $\mathcal{M}^j=\set{{M}^j_{\hat{y}}, \hat{y}\in \mathcal{Y}}$ belongs to $\mathcal{C}$, and $\alpha_j\in [0,1]$ with $\sum_j \alpha_j=1$. By definition, $\Cconvex$ is a convex set. Therefore, it must have extreme points.

\begin{definition} \label{def:C extreme}
Given any concept class $\mathcal{C}$, by $\Cextreme$ denote the extreme points of its convex closure.  
\end{definition}
With this perspective, QSRM is a brute-force search inside $\Cextreme$ and is equiped with the shadow empirical mean estimator of \eqref{eq:shadow loss}.
\addvR{This is summarized as Algorithm~\ref{alg:Shadow QERM} followed by Theorem~\ref{thm:QERM} proving its QPAC learnability.}

\begin{algorithm}[h]
\caption{QSRM}
\label{alg:Shadow QERM}
\DontPrintSemicolon
\KwIn{$\Cextreme$ of the concept class and $n$ training samples.}
  
    \For{$i=1$ \KwTo $n$}{
   		Generate a unitary $U_i$ randomly.\;
   		Apply $U_i$ on $\rho_i$ as in Figure \ref{fig:Estimation}.\;
  		Measure along $\set{\ket{j}, j\in [\dimH]}$  to get $j_i$.\;
            Store the classical shadows $\hat{\rho}_i$ as in \eqref{eq:rho hat shadow}.\;
    }
	\For{each $\mathcal{M}$ in $\Cextreme$}{
 		Compute the empirical loss $L_{\hat{\mathcal{S}}_n}(\mathcal{M})$ as in \eqref{eq:shadow loss}.\;
 		}
 \Return $\hat{\mathcal{M}}$ with the minimum $L_{\hat{\mathcal{S}}_n}$.\;
\end{algorithm}

\begin{theorem}\label{thm:QERM}
Suppose $\ell$ is a bounded loss function and  $\mathcal{C}$ is a measurement class with finite extreme points.  Then, QSRM (Algorithm \ref{alg:Shadow QERM}) agnostically QPAC learns $\mathcal{C}$  with quantum sample complexity bounded as
\begin{align*}
n_{\mathcal{C}}(\epsilon, \delta) = \order{\frac{V_{\Cextreme}}{\epsilon^2}\log\frac{|\Cextreme| }{\delta}},
\end{align*}
where $\mathcal{C}^*$ is the set of extreme points as in Definition \ref{def:C extreme} and
\[V_{\Cextreme}:= \max_{\mathcal{M}\in \Cextreme} \max_{y}\norm{\MLoss(y)-\frac{\tr{\MLoss(y)}}{\dimH} I}^2_{\emph{shadow}},\]
is the shadow norm as in Definition \ref{def:shadow norm}. 
\end{theorem}

\addvR{Note that $|\mathcal{C}^*|\leq |\mathcal{C}|$. Moreover,  $\Cextreme$ can be finite even when $\mathcal{C}$ is infinite. This leads to an interesting distinction compared to the classical learning:  the quantum sample complexity of $\mathcal{C}$ can scale with a rate smaller than  $\log |\mathcal{C}|$.}  It is worth noting that even though the Hilbert space is finite-dimensional, $\Cextreme$ could be infinite. In that case, one can combine QSRM with an $\epsilon$-netting procedure to get a bound on the sample complexity. 

The theorem implies a tight sample complexity bound, up to a constant, for classes of bounded norm. 
\begin{corollary}
    The sample complexity of any concept class $\clC$, such that  $|\Cextreme|$ is finite and the Hilbert–Schmidt norm of each $\mathcal{M}\in \clC$ is bounded by an independent constant, scales as $\Theta( \log |\Cextreme|)$. 
\end{corollary}
\begin{proof}
    The upper bounds follows from Theorem \ref{thm:QERM} and \eqref{eq:shadow norm bound}. The lower bound is implied from the fact that quantum sample complexity is not smaller than classical one, and that the classical sample complexity of a finite concept class $\clC_\emph{classical}$ scales as $\Theta(\log |\clC_\emph{classical}|)$.
\end{proof}



\section{Proof of the main results}
\subsection{Proof of Proposition \ref{prop:shadow empirical estimator}}\label{subsec:proof prop}

The proof  follows from a concentration of measures for bounded variances  \cite[Theorem 8.2]{Dubhashi2009}:

\begin{lem}[Method of Bounded Variances]\label{lem:bounded variances}
Let $X_1, . . ., X_n$ be a set of random variables and let $Y_n = f(X_1, ... , X_n)$, where $f$ a function such that  $\EE[Y_n]<\infty$.
Let $D_i=\EE[Y_n|\bfX^i]-\EE[Y_n|\bfX^{i-1}]$ and $|D_i|\leq c_i$ for some constants $c_i>0$. Also let,  $V_n=\sum_{i=1}^n \sup_{\bfx^{i-1}} \var(D_i|\bfx^{i-1})$. Then, 
\begin{align*}
\prob{\abs{Y_n - \EE[Y_n]}>\epsilon }\leq 2 \exp{-\frac{\epsilon^2}{4V_n}},
\end{align*}
where $\epsilon\leq \frac{2V_n}{\max_i c_i}$.
\end{lem}

\begin{corollary}
For any $O$, and small enough $\epsilon>0$,  the shadow empirical loss in \eqref{eq:shadow loss} satisfies
\begin{align*}
\prob{\abs{\hat{o}  - \tr{O\rho}}> \epsilon } \leq 2\exp{\frac{-n\epsilon^2}{4\max_i \var(\tr{O\hat{\rho}_i})}}.
\end{align*}
\end{corollary}
\begin{proof}
The corollary is proved from Lemma \ref{lem:bounded variances} with $f=\hat{o} := \frac{1}{n}\sum_{i}\tr{O\hat{\rho}_i}$ implying that  
$D_i=\frac{1}{n} (\tr{O\hat{\rho}_i}-\EE[\tr{O\hat{\rho}_i}]).$
Since, $\hat{\rho}_i$'s are mutually independent and identically distributed then 
$V_n=\frac{1}{n^2}\sum_i \var(\tr{O\hat{\rho}_i}) \leq \frac{1}{n}\max_i \var(\tr{O\hat{\rho}_i}),$
which gives the desired  statement.
\end{proof}

Applying this result for $O=O_i$ in the proposition and a union bound give 
$\max_j \abs{\hat{o}_j  - \tr{O_j\rho}} \leq \varepsilon$
 with probability greater than
 \begin{align}\label{eq:prob 1}
 1 - 2m \exp{\frac{-n\epsilon^2}{4 \max_{i,j} \var(\tr{O_j\hat{\rho}_i})}}.
 \end{align}
From \cite[Lemma 1]{Huang2020}, the variance terms are upper bounded by the shadow norm as 
$\max_{j} \norm{O_j-\frac{\tr{O_j}}{\dimH} I}^2_{\emph{shadow}}.$
Equating the bound on the probability to $\delta$ gives the desired result. 

%

\subsection{Proof of Theorem \ref{thm:QERM}.} 

We start with the following lemma: 
\begin{lem}\label{lem:extreme points}
Let $\Cconvex$ be the convex closure of $\mathcal{C}$ and $\Cextreme$ be the set of all extreme points of $\Cconvex$. Then, 
$\opt_{\mathcal{C}}=\opt_{\Cconvex}=\inf_{\mathcal{M}\in \Cextreme} \Loss(\mathcal{M}).$
\end{lem}
\begin{proof}
    Note that $\Loss(\mathcal{M})$ is linear in $\mathcal{M}$ as it is equal to the expectation of the loss observable $\MLoss$. This is due to the linearity of the trace and the definition of $\MLoss$.  As a result, given that $\mathcal{C}\subseteq \Cconvex$ and that
       $\opt_{\Cconvex} =\inf_{\mathcal{M}\in \Cconvex} \Loss(\mathcal{M})$
    we find that $ \opt_{\Cconvex}= \opt_{\mathcal{C}}$. Moreover, since the above expression is a convex optimization, then the optimal values occur at the extreme points of $\Cconvex$. Hence the proof is complete. 
\end{proof}

This result implies that QPAC learning of $\mathcal{C}$ is reduced to its extreme points $\Cextreme$. 
The proof of Theorem \ref{thm:QERM} follows from Proposition~\ref{prop:shadow empirical estimator} and Lemma \ref{lem:extreme points} with the following argument.
For any fixed $y$, and $\mathcal{M}$ let $\MLoss(y)$ be the expected loss when the label is $y$. We apply Proposition~\ref{prop:shadow empirical estimator} with $O_j = \MLoss(y)$ for $\mathcal{M}\in \Cextreme$. Let 
\[V_{\Cextreme}:= \max_{\mathcal{M}\in \Cextreme} \max_{y}\norm{\MLoss(y)-\frac{\tr{\MLoss(y)}}{\dimH} I}^2_{\emph{shadow}}.\]
Then, the proposition  gives the following sample complexity bound: 
$n  = \order{\frac{V_{\Cextreme}}{\epsilon^2}\log(\frac{|\Cextreme| }{\delta})}.$
Now, let $\widehat{\mathcal{M}}$ and $\mathcal{M}^*$ be the measurements minimizing  $L_{\hat{\mathcal{S}}_n}$ and $\Loss$, respectively.  Then,  with probability $(1-\delta)$ we have that:
\begin{align*}
\Loss(\widehat{\mathcal{M}}) \leq \hat{L}(\widehat{\mathcal{M}})+\frac{\epsilon}{2} 
&\leq  \hat{L}({\mathcal{M}}^*)+\frac{\epsilon}{2} \leq \Loss({\mathcal{M}}^*) +\epsilon.
\end{align*}
The left-hand side is the loss of the selected predictor by QSRM (Algorithm \ref{alg:Shadow QERM}), and the right-hand side equals $\opt_{\mathcal{C}}+\epsilon$ and hence the proof is complete.

\section*{Conclusion}
This paper studies the learning of quantum measurement classes. It introduces a novel quantum algorithm called QSRM for learning  quantum concept classes. Using this algorithm,  a new upper bound on the quantum sample complexity is derived. It is shown that the quantum sample complexity grows at most with the logarithm of the size of the extreme points of the convex closure of the concept class. This is a significant improvement over prior results. The approach is based on a novel method to estimate the empirical loss of the concept class via creating random shadows of the training samples. With that QSRM algorithm can perform risk minimization while abiding to no-cloning, state collapse, and measurement incompatibility.

 \section*{Acknowledgment}

 This work was partially supported by
 the NSF Center for Science of Information (CSoI)
 Grant CCF-0939370, and in addition by NSF Grants CCF-1524312, CCF-2006440,
 and CCF-2211423.

 \bibliographystyle{IEEEtran}
\bibliography{main}
\end{document}

%% file: PackageHeader2022.tex
\usepackage{amsmath}
\usepackage{amsfonts}
\usepackage{amsthm}   
\usepackage{mathrsfs}
\usepackage{tikz}
\usepackage[utf8]{inputenc}

\usepackage{epstopdf}
\usepackage{mathtools}
\usepackage{graphicx}
\usepackage{bbm}
\usepackage{enumerate}
\usepackage{epsf,verbatim,amssymb,array,cite,multicol,multirow}  
\usepackage{psfrag,bm,xspace}
\usepackage{hhline}
\usepackage{xstring}
\usepackage{ifthen}
\usepackage{booktabs}

\usepackage{xparse}
\usepackage{physics}

\usepackage[linesnumbered,ruled,vlined,procnumbered]{algorithm2e}

\ifdefined \theorem 
\else
  \newtheorem{theorem}{Theorem}
\fi
\newtheorem{lem}{Lemma}

\ifdefined \corollary 
\else
\newtheorem{corollary}{Corollary}
\fi

\newtheorem{fact}{Fact}

\ifdefined \proposition 
\else
\newtheorem{proposition}{Proposition}
\fi

\ifdefined \definition 
\else
\newtheorem{definition}{Definition}
\fi

\ifdefined \example 
\else
\newtheorem{example}{Example}
\fi

\ifdefined \remark 
\else
\newtheorem{remark}{Remark}
\fi

\ifdefined \conjecture 
\else
\newtheorem{conjecture}{Conjecture}
\fi

\makeatletter
\def\old@comma{,}
\catcode`\,=13
\def,{%
  \ifmmode%
    \old@comma\discretionary{}{}{}%
  \else%
    \old@comma%
  \fi%
}
\makeatother

\usepackage{xcolor}
\definecolor{darkblue}{rgb}{0.1,0.1,0.8}
\definecolor{DarkGreen}{rgb}{0,0.6,0}
\definecolor{brickred}{rgb}{0.8, 0.25, 0.33}
\definecolor{britishracinggreen}{rgb}{0.0, 0.26, 0.15}
\definecolor{calpolypomonagreen}{rgb}{0.12, 0.3, 0.17}
\definecolor{ao(english)}{rgb}{0.0, 0.5, 0.0}
	\definecolor{cadmiumgreen}{rgb}{0.0, 0.42, 0.24}
\definecolor{burgundy}{rgb}{0.5, 0.0, 0.13}
\usepackage{etoolbox}

\providetoggle{Blue_revision}
\settoggle{Blue_revision}{true}

\providetoggle{Track}
\settoggle{Track}{true}
\newcommand{\addv}[3]{%
	\iftoggle{Track}{%
    	\IfEqCase{#1}{%
       	 	{a}{\ifthenelse{\equal{#2}{ON}}{{\color{cadmiumgreen}#3}}{#3}}%
        	{b}{\ifthenelse{\equal{#2}{ON}}{{\color{brickred}#3}}{#3}}%
       		{c}{\ifthenelse{\equal{#2}{ON}}{{\color{burgundy}#3}}{#3}}%
    	}[\PackageError{tree}{Undefined option to tree: #1}{}]%
	}{#3}%
}

 \usepackage{hyperref}

\hypersetup{
colorlinks=true, %
  pdfstartview={FitH},
    linkcolor=red,
    citecolor=blue,
    urlcolor={blue!80!black}
}

\usepackage{graphicx}

\newcounter{relctr} 
\everydisplay\expandafter{\the\everydisplay\setcounter{relctr}{0}} 

\AtBeginDocument{} 

\newcommand{\etal}{\textit{et al. }}

\global\long\def\RR{\mathbb{R}}

\global\long\def\NN{\mathbb{N}}

\global\long\def\EE{\mathbb{E}}
\global\long\def\PP{\mathbb{P}}

\global\long\def\11{\mathbbm{1}}

\newcommand{\bfs}{\mathbf{s}}

\newcommand{\bfx}{\mathbf{x}}

\newcommand{\bfX}{\mathbf{X}}

\newcommand{\clC}{\mathcal{C}}

\global\long\def\+{\oplus}

\newcommand{\prob}[1]{\PP\Big\{  #1 \Big\} }
\def\<{\langle}
\def\>{\rangle}

\ifdefined \var 
  \renewcommand{\var}{\mathsf{var}}
\else
  \newcommand{\var}{\mathsf{var}}
\fi

\ifdefined \abs \else
 \newcommand{\abs}[1]{\lvert#1\rvert}
\fi

\ifdefined \norm \else
 \newcommand{\norm}[1]{\lVert#1\rVert}
\fi

\ifdefined \set 
  \renewcommand{\set}[1]{\left\{#1\right\}}
\else
  \newcommand{\set}[1]{\left\{#1\right\}}
\fi

\usepackage{stackengine}

%% file: Notations22.tex
\DeclareMathOperator*{\tensor}{\otimes}

\providecommand{\tr}{tr}

\ifdefined \Tr 
  \renewcommand{\Tr}[1]{\tr \Big\{#1\Big\}}
\else
  \newcommand{\Tr}[1]{\tr \Big\{#1\Big\}}
\fi

\def\MLoss{\mathcal{L}_M}
\def\Loss{L_{D}}

\def\Cconvex{\bar{\mathcal{C}}}
\def\Cextreme{\mathcal{C}^*}
\def\dimH{\text{dim}({\mathcal{H}})}

\def\qps{\sigma^\bfs}

\newcommand{\opt}{\optfont{opt}}
\newcommand{\optfont}[1]{\mathsf{#1}}

%% file: acronyms.tex
\usepackage[nolist,nohyperlinks]{acronym}

\newacro{ptp}[PtP]{Point-to-Point}
\newacro{iid}[i.i.d.]{independent and identically distributed} 
\newacro{IID}[i.i.d.]{independent and identically distributed} 
\newacro{PAC}[PAC]{\textit{probably approximately correct}}
\newacro{VC}[VC]{Vapnik–Chervonenkis}
\newacro{ERM}[ERM]{\textit{empirical risk minimization}}
\newacro{SVM}[SVM]{support-vector machine}
\newacro{POVM}[POVM]{positive operator-valued measure}
\newacro{QPAC}[QPAC]{\textit{quantum probably approximately correct}}

\newacro{QSRM}[QSRM]{\textit{quantum shadow risk minimization}}
\newacro{QC}[QC]{quantum computer}
\acrodefplural{OC}[OC's]{quantum computers}
\newacro{ML}[ML]{machine learning}
\newacro{QML}[QML]{quantum-enhanced machine learning}
\newacro{NISQ}[NISQ]{noisy intermediate-scale quantum}
\newacro{VQA}[VQA]{variational quantum algorithms}
\newacro{QNN}[QNN]{quantum neural network}

%% file: main.bbl
\begin{thebibliography}{10}
\providecommand{\url}[1]{#1}
\csname url@samestyle\endcsname
\providecommand{\newblock}{\relax}
\providecommand{\bibinfo}[2]{#2}
\providecommand{\BIBentrySTDinterwordspacing}{\spaceskip=0pt\relax}
\providecommand{\BIBentryALTinterwordstretchfactor}{4}
\providecommand{\BIBentryALTinterwordspacing}{\spaceskip=\fontdimen2\font plus
\BIBentryALTinterwordstretchfactor\fontdimen3\font minus
  \fontdimen4\font\relax}
\providecommand{\BIBforeignlanguage}[2]{{%
\expandafter\ifx\csname l@#1\endcsname\relax
\typeout{** WARNING: IEEEtran.bst: No hyphenation pattern has been}%
\typeout{** loaded for the language `#1'. Using the pattern for}%
\typeout{** the default language instead.}%
\else
\language=\csname l@#1\endcsname
\fi
#2}}
\providecommand{\BIBdecl}{\relax}
\BIBdecl

\bibitem{Giovannetti2008}
V.~Giovannetti, S.~Lloyd, and L.~Maccone, ``Quantum random access memory,''
  \emph{Physical Review Letters}, vol. 100, no.~16, apr 2008.

\bibitem{Park2019}
D.~K. Park, F.~Petruccione, and J.-K.~K. Rhee, ``Circuit-based quantum random
  access memory for classical data,'' \emph{Scientific Reports}, vol.~9, no.~1,
  mar 2019.

\bibitem{Lloyd2013}
S.~Lloyd, M.~Mohseni, and P.~Rebentrost, ``Quantum algorithms for supervised
  and unsupervised machine learning,'' \emph{arXiv:1307.0411}, 2013.

\bibitem{Lloyd2014}
------, ``Quantum principal component analysis,'' vol.~10, no.~9, pp. 631--633,
  jul 2014.

\bibitem{Carrasquilla2017}
J.~Carrasquilla and R.~G. Melko, ``Machine learning phases of matter,''
  vol.~13, no.~5, pp. 431--434, feb 2017.

\bibitem{Carleo2017}
G.~Carleo and M.~Troyer, ``Solving the quantum many-body problem with
  artificial neural networks,'' vol. 355, no. 6325, pp. 602--606, feb 2017.

\bibitem{Broughton2020}
M.~Broughton, G.~Verdon, T.~McCourt, A.~J. Martinez, J.~H. Yoo, S.~V. Isakov,
  P.~Massey, R.~Halavati, M.~Y. Niu, A.~Zlokapa, E.~Peters, O.~Lockwood,
  A.~Skolik, S.~Jerbi, V.~Dunjko, M.~Leib, M.~Streif, D.~V. Dollen, H.~Chen,
  S.~Cao, R.~Wiersema, H.-Y. Huang, J.~R. McClean, R.~Babbush, S.~Boixo,
  D.~Bacon, A.~K. Ho, H.~Neven, and M.~Mohseni, ``Tensorflow quantum: A
  software framework for quantum machine learning,'' \emph{arXiv:2003.02989},
  Mar. 2020.

\bibitem{massoli2021leap}
F.~V. Massoli, L.~Vadicamo, G.~Amato, and F.~Falchi, ``A leap among
  entanglement and neural networks: A quantum survey,''
  \emph{arXiv:2107.03313}, Jul. 2021.

\bibitem{Lu2018}
S.~Lu, S.~Huang, K.~Li, J.~Li, J.~Chen, D.~Lu, Z.~Ji, Y.~Shen, D.~Zhou, and
  B.~Zeng, ``Separability-entanglement classifier via machine learning,''
  \emph{Physical Review A}, vol.~98, no.~1, p. 012315, 2018.

\bibitem{Kassal2011}
I.~Kassal, J.~D. Whitfield, A.~Perdomo-Ortiz, M.-H. Yung, and A.~Aspuru-Guzik,
  ``Simulating chemistry using quantum computers,'' \emph{Annual Review of
  Physical Chemistry}, vol.~62, no.~1, pp. 185--207, may 2011.

\bibitem{McArdle2020}
S.~McArdle, S.~Endo, A.~Aspuru-Guzik, S.~C. Benjamin, and X.~Yuan, ``Quantum
  computational chemistry,'' \emph{Reviews of Modern Physics}, vol.~92, no.~1,
  p. 015003, mar 2020.

\bibitem{Cao2019}
Y.~Cao, J.~Romero, J.~P. Olson, M.~Degroote, P.~D. Johnson, M.~Kieferov{\'{a}},
  I.~D. Kivlichan, T.~Menke, B.~Peropadre, N.~P.~D. Sawaya, S.~Sim, L.~Veis,
  and A.~Aspuru-Guzik, ``Quantum chemistry in the age of quantum computing,''
  \emph{Chemical Reviews}, vol. 119, no.~19, pp. 10\,856--10\,915, aug 2019.

\bibitem{Hempel2018}
C.~Hempel, C.~Maier, J.~Romero, J.~McClean, T.~Monz, H.~Shen, P.~Jurcevic,
  B.~P. Lanyon, P.~Love, R.~Babbush, A.~Aspuru-Guzik, R.~Blatt, and C.~F. Roos,
  ``Quantum chemistry calculations on a trapped-ion quantum simulator,''
  \emph{Physical Review X}, vol.~8, no.~3, p. 031022, jul 2018.

\bibitem{Bauer2020}
B.~Bauer, S.~Bravyi, M.~Motta, and G.~K.-L. Chan, ``Quantum algorithms for
  quantum chemistry and quantum materials science,'' \emph{Chemical Reviews},
  vol. 120, no.~22, pp. 12\,685--12\,717, oct 2020.

\bibitem{Broecker2017}
P.~Broecker, J.~Carrasquilla, R.~G. Melko, and S.~Trebst, ``Machine learning
  quantum phases of matter beyond the fermion sign problem,'' vol.~7, no.~1,
  aug 2017.

\bibitem{Biamonte2017}
J.~Biamonte, P.~Wittek, N.~Pancotti, P.~Rebentrost, N.~Wiebe, and S.~Lloyd,
  ``Quantum machine learning,'' vol. 549, no. 7671, pp. 195--202, sep 2017.

\bibitem{Ma2018}
Y.-C. Ma and M.-H. Yung, ``Transforming bell's inequalities into state
  classifiers with machine learning,'' \emph{npj Quantum Information}, vol.~4,
  no.~1, jul 2018.

\bibitem{Hiesmayr2021}
B.~C. Hiesmayr, ``Free versus bound entanglement, a {NP}-hard problem tackled
  by machine learning,'' \emph{Scientific Reports}, vol.~11, no.~1, oct 2021.

\bibitem{Chen2021a}
C.~Chen, C.~Ren, H.~Lin, and H.~Lu, ``Entanglement structure detection via
  machine learning,'' \emph{Quantum Science and Technology}, 2021.

\bibitem{Deng2017}
D.-L. Deng, X.~Li, and S.~D. Sarma, ``Quantum entanglement in neural network
  states,'' \emph{Physical Review X}, vol.~7, no.~2, p. 021021, 2017.

\bibitem{Kivlichan2020}
I.~D. Kivlichan, C.~Gidney, D.~W. Berry, N.~Wiebe, J.~McClean, W.~Sun,
  Z.~Jiang, N.~Rubin, A.~Fowler, A.~Aspuru-Guzik, H.~Neven, and R.~Babbush,
  ``Improved fault-tolerant quantum simulation of condensed-phase correlated
  electrons via trotterization,'' \emph{Quantum}, vol.~4, p. 296, jul 2020.

\bibitem{Kearns1994}
M.~J. Kearns, R.~E. Schapire, and L.~M. Sellie, ``Toward efficient agnostic
  learning,'' \emph{Machine Learning}, vol.~17, no. 2-3, pp. 115--141, 1994.

\bibitem{Valiant1984}
L.~G. Valiant, ``A theory of the learnable,'' \emph{Communications of the
  {ACM}}, vol.~27, no.~11, pp. 1134--1142, nov 1984.

\bibitem{Aaronson2007}
S.~Aaronson, ``The learnability of quantum states,'' \emph{Proceedings of the
  Royal Society A: Mathematical, Physical and Engineering Sciences}, vol. 463,
  no. 2088, pp. 3089--3114, sep 2007.

\bibitem{Barnett2009}
S.~M. Barnett and S.~Croke, ``Quantum state discrimination,'' \emph{Advances in
  Optics and Photonics}, vol.~1, no.~2, p. 238, feb 2009.

\bibitem{Bshouty1998}
N.~H. Bshouty and J.~C. Jackson, ``Learning dnf over the uniform distribution
  using a quantum example oracle,'' \emph{SIAM Journal on Computing}, vol.~28,
  no.~3, pp. 1136--1153, 1998.

\bibitem{Badescu2019}
C.~Badescu, R.~O'Donnell, and J.~Wright, ``Quantum state certification,'' in
  \emph{Proceedings of the 51st Annual {ACM} {SIGACT} Symposium on Theory of
  Computing}.\hskip 1em plus 0.5em minus 0.4em\relax {ACM}, jun 2019.

\bibitem{Montanaro2016}
A.~Montanaro and R.~de~Wolf, ``A survey of quantum property testing,''
  \emph{Theory of Computing}, vol.~1, no.~1, pp. 1--81, 2016.

\bibitem{HeidariQuantum2021}
M.~Heidari, A.~Padakandla, and W.~Szpankowski, ``A theoretical framework for
  learning from quantum data,'' in \emph{{IEEE} International Symposium on
  Information Theory ({ISIT})}, 2021.

\bibitem{Heidari2023Quantum1}
M.~Heidari and W.~Szpankowski, ``Learning k-qubit quantum operators via pauli
  decomposition,'' in \emph{Proceedings of The 26th International Conference on
  Artificial Intelligence and Statistics}, ser. Proceedings of Machine Learning
  Research, F.~Ruiz, J.~Dy, and J.-W. van~de Meent, Eds., vol. 206.\hskip 1em
  plus 0.5em minus 0.4em\relax PMLR, 25--27 Apr 2023.

\bibitem{ODonnell2016}
R.~O'Donnell and J.~Wright, ``Efficient quantum tomography,'' in
  \emph{Proceedings of the forty-eighth annual {ACM} symposium on Theory of
  Computing}.\hskip 1em plus 0.5em minus 0.4em\relax {ACM}, jun 2016.

\bibitem{Haah2016}
J.~Haah, A.~W. Harrow, Z.~Ji, X.~Wu, and N.~Yu, ``Sample-optimal tomography of
  quantum states,'' in \emph{Proceedings of the forty-eighth annual {ACM}
  symposium on Theory of Computing}.\hskip 1em plus 0.5em minus 0.4em\relax
  {ACM}, jun 2016.

\bibitem{Bubeck2020}
S.~Bubeck, S.~Chen, and J.~Li, ``Entanglement is necessary for optimal quantum
  property testing,'' in \emph{2020 {IEEE} 61st Annual Symposium on Foundations
  of Computer Science ({FOCS})}.\hskip 1em plus 0.5em minus 0.4em\relax {IEEE},
  nov 2020.

\bibitem{Gambs2008}
S.~Gambs, ``Quantum classification,'' \emph{0809.0444 [quant-ph]}, Sep. 2008.

\bibitem{Guta2010}
M.~Guta and W.~Kotlowski, ``Quantum learning: asymptotically optimal
  classification of qubit states,'' \emph{New Journal of Physics}, vol.~12,
  no.~12, p. 123032, dec 2010.

\bibitem{Cheng2021}
H.-C. Cheng, A.~Winter, and N.~Yu, ``Discrimination of quantum states under
  locality constraints in the many-copy setting,'' in \emph{2021 {IEEE}
  International Symposium on Information Theory ({ISIT})}.\hskip 1em plus 0.5em
  minus 0.4em\relax {IEEE}, jul 2021.

\bibitem{Arunachalam2017}
S.~Arunachalam and R.~de~Wolf, ``A survey of quantum learning theory,''
  \emph{arXiv:1701.06806}, 2017.

\bibitem{Kanade2018}
V.~Kanade, A.~Rocchetto, and S.~Severini, ``Learning dnfs under product
  distributions via $\mu$-biased quantum fourier sampling,''
  \emph{arXiv:1802.05690v3}, 2019.

\bibitem{Bernstein1997}
E.~Bernstein and U.~Vazirani, ``Quantum complexity theory,'' \emph{{SIAM}
  Journal on Computing}, vol.~26, no.~5, pp. 1411--1473, oct 1997.

\bibitem{Servedio2004}
R.~A. Servedio and S.~J. Gortler, ``Equivalences and separations between
  quantum and classical learnability,'' \emph{SIAM J. Comput.}, vol.~33, no.~5,
  p. 1067–1092, May 2004.

\bibitem{Padakandla2022}
A.~Padakandla and A.~Magner, ``Pac learning of quantum measurement classes :
  Sample complexity bounds and universal consistency,'' in \emph{Proceedings of
  The 25th International Conference on Artificial Intelligence and Statistics},
  ser. Proceedings of Machine Learning Research, G.~Camps-Valls, F.~J.~R. Ruiz,
  and I.~Valera, Eds., vol. 151.\hskip 1em plus 0.5em minus 0.4em\relax PMLR,
  28--30 Mar 2022, pp. 11\,305--11\,319.

\bibitem{HeidariQuantum2019}
M.~Heidari, T.~A. Atif, and S.~S. Pradhan, ``Faithful simulation of distributed
  quantum measurements with applications in distributed rate-distortion
  theory,'' in \emph{2019 {IEEE} International Symposium on Information Theory
  ({ISIT})}.\hskip 1em plus 0.5em minus 0.4em\relax {IEEE}, jul 2019.

\bibitem{Huang2020}
H.-Y. Huang, R.~Kueng, and J.~Preskill, ``Predicting many properties of a
  quantum system from very few measurements,'' \emph{Nature Physics 16,
  1050--1057 (2020)}, Feb. 2020.

\bibitem{Holevo2012}
A.~S. Holevo, \emph{Quantum Systems, Channels, Information}.\hskip 1em plus
  0.5em minus 0.4em\relax {DE} {GRUYTER}, jan 2012.

\bibitem{Helsen2023}
J.~Helsen and M.~Walter, ``Thrifty shadow estimation: Reusing quantum circuits
  and bounding tails,'' \emph{Physical Review Letters}, vol. 131, no.~24, p.
  240602, Dec. 2023.

\bibitem{Zhao2021}
A.~Zhao, N.~C. Rubin, and A.~Miyake, ``Fermionic partial tomography via
  classical shadows,'' \emph{Physical Review Letters}, vol. 127, no.~11, p.
  110504, Sep. 2021.

\bibitem{Dubhashi2009}
D.~P. Dubhashi and A.~Panconesi, \emph{Concentration of Measure for the
  Analysis of Randomized Algorithms}.\hskip 1em plus 0.5em minus 0.4em\relax
  Cambridge University Press, Jun. 2009.

\end{thebibliography}
